\documentclass[11pt,reqno,a4paper]{amsart}
\usepackage{amsmath,amssymb}
\usepackage[english,activeacute]{babel}

\usepackage[active]{srcltx}
 \newtheorem{thm}{Theorem}[section]
 
 \newtheorem{cor}[thm]{Corollary}
 \newtheorem{lem}[thm]{Lemma}
 \newtheorem{prop}[thm]{Proposition}
 \theoremstyle{definition}
 \newtheorem{defn}[thm]{Definition}
 
 \theoremstyle{remark}
 \newtheorem{rem}[thm]{Remark}

 \numberwithin{equation}{section}

\newcommand{\be}{\begin{equation}}
\newcommand{\ee}{\end{equation}}
\newcommand{\benon}{\begin{equation*}}
\newcommand{\eenon}{\end{equation*}}
\newcommand{\ba}{\begin{array}}
\newcommand{\ea}{\end{array}}
\newcommand{\bal}{\begin{align}}
\newcommand{\eal}{\end{align}}
\newcommand{\bea}{\begin{eqnarray}}
\newcommand{\eea}{\end{eqnarray}}
\newcommand{\bee}{\begin{eqnarray*}}
\newcommand{\eee}{\end{eqnarray*}}
\renewcommand{\P}{\mathbb P}
\newcommand{\Z}{\mathbb Z}
\newcommand{\Zd}{\mathbb Z^d}
\newcommand{\Rd}{{\mathbb R^d}}
\newcommand{\R}{\mathbb R}

\newcommand{\N}{\mathbb N}

\newcommand{\Prob}[1]{\mathbb P\left(#1\right)}
\newcommand{\norm}[1]{\Vert #1 \Vert}
\newcommand{\abs}[1]{\left| #1 \right|}
\newcommand{\Lp}[1]{\textrm{L}^2(#1)}
\renewcommand{\L}{{\Lambda}}

\newcommand{\angles}[1]{\langle #1 \rangle}
\numberwithin{equation}{section}
\newcommand{\non}{\nonumber}

\numberwithin{equation}{section}

\newcommand{\hm}[1]{\leavevmode{\marginpar{\tiny%
$\hbox to 0mm{\hspace*{-0.5mm}$\leftarrow$\hss}%
\vcenter{\vrule depth 0.1mm height 0.1mm width \the\marginparwidth}%
\hbox to 0mm{\hss$\rightarrow$\hspace*{-0.5mm}}$\\\relax\raggedright #1}}}

\newcommand{\dist}{\mathop{\mathrm{dist}}}

\newcommand{\supp}{{\mathop{\mathrm{supp\,}}}}
\newcommand{\LL}{{\Lambda_L}}
\renewcommand{\L}{\Lambda}

\newcommand{\Hom}{H_\omega}
\newcommand{\Vom}{V_\omega}

\newcommand{\om}{\omega}

\newcommand{\ran}{\textrm{Ran }}

\begin{document}

\title[The Anderson model with missing sites]{The Anderson model with missing sites}

\author{Constanza Rojas-Molina}
\address{Universit\'e de Cergy-Pontoise, UMR CNRS 8088, F-95000 Cergy-Pontoise, France and CAMTP - Center for Applied Mathematics and Theoretical Physics, University of Maribor, Krekova 2, SI-2000 Maribor,
Slovenia.}
\email{crojasm@u-cergy.fr}



\begin{abstract}
In the present note we show dynamical localization for an Anderson model with missing sites in a discrete setting at the bottom of the spectrum in arbitrary dimension $d$. In this model, the random potential is defined on a relatively dense subset of $\Zd$, not necessarily periodic, i.e., a Delone set in $\Z^d$. To work in the lower band edge we need no further assumption on the geometric complexity of the Delone set. We use a spatial averaging argument by Bourgain-Kenig to obtain a uniform Wegner estimate and an initial length scale estimate, which yields localization through the Multiscale Analysis for non ergodic models. This argument gives an explicit dependence on the maximal distance parameter of the Delone set for the Wegner estimate. We discuss the case of the upper spectral band edge and the arising need of imposing the (complexity) condition of strict uniform pattern frequency on the Delone set.
\end{abstract}

\maketitle

\section{Introduction and main result}

Consider the operator
\be\label{formula-ranop} \Hom=-\Delta+ \Vom \quad \mbox{ on }\,\,\ell^2(\Zd)\ee
where $d\geq 1$ is the dimension,
\be\label{formula-ranpot}
   \Vom(n)=\left\{ \begin{array}{ll}
\omega(n) & \textrm{if $n\in D$}\\
0 & \textrm{otherwise}.\\
\end{array} \right.
\ee
where $(\omega(n))_{n\in D}$ are i.i.d. random variables with regular, compactly supported probability distribution $\mu$ such that $\supp \mu\subset [0,M]$, for $M>0$. We denote the probability space by $(\Omega, \P)$, where $\Omega=[0,M]^D$. Here $D$ is a subset of $\Zd$, not necessarily periodic, and relatively dense. This means there exists a constant $R\geq 1$, such that every closed cube in $\Zd$ of side length $R$ in the max-norm of $\Zd$ contains at least one point of $D$. We call such $D$ an $R$-\emph{Delone set} in $\Zd$ (note that the usual uniform discreteness property used to define Delone sets is trivially satisfied in $\Zd$). The operator $\Hom$ with the potential $\Vom$ is called the Anderson model with missing sites and fits in the more general framework of Delone--Anderson operators \cite{RM12, GMRM, BdMNSS06,RMV12, Kl12}.

We say that $H_\omega$ exhibits \emph{dynamical localization} in an interval $I$, $I\cap\sigma(H_\omega)\neq \emptyset$, if for any $\psi\in \ell^2(\Zd)$ we have, with probability one:

\be\label{def-DL} \sup_{t\in\R} \norm{\angles{X}^{p/2} e^{-itH_\omega}P_\omega(I)\psi} <\infty,\quad \mbox{for all } p\geq 0, \ee
where $P_\omega(I)$ denotes the spectral projection of $H_\omega$ associated to the interval $I$. Note that the norm in \eqref{def-DL} is a measure of the spreading of the wave $\psi$ in space as a function of time. Therefore, dynamical localization means absence of diffusion, in the sense that the particle stays localized in space uniformly for all times. For a thorough exposition of the subject, see \cite{K07,Kl08,S}. The method we use to prove localization, the Multiscale Analysis, actually proves stronger notions of dynamical localization, see the discussion in \cite[Section 3]{Kl08}.

For standard Anderson models, i.e. with $D=\Zd$, dynamical localization has been proved at the bottom of the spectrum in arbitrary dimension.
 The Anderson model with missing sites, on the other hand, presents the difficulty of having gaps in the potential, so there is no covering condition. In the continuum, the problem of obtaining dynamical localization and proving Wegner estimates without any covering condition was treated in \cite{CHK03, KV02a,RMV12,Kl12}, among others.

 More general models, but still associated to periodic settings in $\Zd$ were considered in \cite{CaoE11} in the Lifshitz tails regime at the bottom of the spectrum. This allowed the authors to prove dynamical localization using the Fractional Moment Method \cite{AM93}. However, in the case where $D$ is non periodic, $\Hom$ is a non ergodic Hamiltonian and therefore the Integrated Density of States (IDS) cannot be proven to exist as a non random object with the standard techniques. This problem was addressed for Delone--Anderson models on graphs in \cite[Section 4]{MR07} and in the continous setting in \cite{GMRM}, showing that under certain conditions on the geometric complexity of the Delone set, as finite local complexity and strict uniform pattern frequency, the IDS exists, is non random and it gives information about the almost sure spectrum of $H_\omega$.
The IDS for Delone operators, that is, purely aperiodic settings, has been thoroughly studied in the literature, see e.g. \cite{LS03, LS05, LS06, LV09}.

For a study of dynamical localization in non ergodic models on $\Lp{\Rd}$ and in particular, Delone--Anderson models, see \cite{RM12}. Previous results for the latter model were obtained in \cite{BdMNSS06} using the Fractional Moment Method.  In \cite{RM12} a generalization of the Bootstrap multiscale analysis \cite{GK01} was obtained for continuous non ergodic models, which holds in the discrete setting, under the condition that all probabilistic estimates involved are uniform with respect to translations in space. Therefore, to prove that $\Hom$ exhibits dynamical localization at the spectral band edges, it is enough to show that a uniform Wegner estimate and a uniform initial length scale estimate hold for energies in that range.

The main purpose of this note is to show dynamical localization for $H_\omega$ at the lower band edge, using an averaging argument by Bourgain-Kenig \cite{BoK05} (see also \cite{CHK07,G08,GMRM}). We also show that a geometric assumption on the Delone set $D$ is enough to establish the existence of almost sure spectrum around the upper band edge, or, equivalently, at the bottom of the spectrum in the case of a non trivial bounded background potential. To study the bottom of the spectrum we need no assumption on the existence of the IDS nor we need information on its asymptotic behavior. As a consequence, we do not need to impose any condition on the geometric complexity of $D$. For the upper band edge, however, one needs to impose that $D$ satisfies the property of strict positive uniform pattern frequency (SUPF). Note that since $D$ is embedded in $\Zd$, it is trivially of finite local complexity (FLC) (see \cite[Section 2.3]{MR12} for a complete definition).
These geometric conditions are needed to ensure that there is almost surely spectrum in the region where one can prove localization. Once $D$ satisfies the property SUPF, and since it is of FLC, results from \cite{MR07,GMRM} show that the IDS exists, is non random and the spectrum of $H_\omega$ is given by the support of the density of states measure associated to the IDS. Then, there exists a set $\Sigma\subset \R$, such that $\Sigma=\sigma(H_\omega)$ for almost every $\omega\in\Omega$.

Let us recall that for $u\in \ell^2(\Zd) $,
\be H_0 u(n)= -\Delta u(n) = - \displaystyle\sum_{|m-n|_1=1}\left(u(m)-u(n)\right).  \ee
Using a Fourier transform, it is easy to see that $\sigma(H_0)=[0,4d]$, where $\sigma(A)$ denotes the spectrum of an operator $A$.  By assuming $0\in \supp \mu$, one can use a Borel-Cantelli type argument (see \cite[Section 6.4]{RM12}) to show that
\be [0,4d] \subset \sigma(H_\omega)\subset [0,4d+M], \quad \mbox{almost surely}\ee
The lower band edge of $\sigma(H_\omega)$ corresponds to $E_0=0$, while the upper band edge corresponds to some value $E_*=\sup \Sigma\in [4d, 4d+M]$. In Section \ref{s:upper} we show that, although we do not know if there are spectral gaps in the spectrum beyond $4d$, the energy $E_*$ is not an isolated spectral value of $\sigma(H_\omega)$, almost surely.

Our main result is the following (\cite{RMthesis}):

\begin{thm}\label{thm-main} Let $D$ be an arbitrary Delone set and assume $0\in\supp\mu$. For any $d\geq 1$, $H_\omega$ defined in \eqref{formula-ranop} exhibits dynamical localization at the bottom of the spectrum.
\end{thm}

Our proof relies on a spatial averaging argument by Bourgain-Kenig \cite[Section 4]{BoK05}, see also \cite{GHK07,G08,GMRM}. It consists in using an averaged version of the original potential, that proves to be a good approximation at the bottom of the spectrum of $-\Delta$. The case of the upper band edge is equivalent to study the bottom of the spectrum of an operator $\tilde H_\om=\tilde H_0+\tilde V_\omega$, where $\tilde H_0$ is now a perturbation of the Laplacian by a potential supported in $\Z^d\setminus D$. In this case the averaging argument no longer works.
\begin{rem}
As this note was being written, the author learned about Elgart and Klein's recent work \cite{EKl13}. There, they use another approach to prove dynamical localization at spectral band edges for $\tilde H_\om$ without averaging of the potential, that allows them to treat the case of a non trivial background potential.
 \end{rem}
In our proof, the averaging argument gives concentration estimates on the eigenfunctions of finite volume versions of $H_\omega$. Basically, we prove that a Delone potential accounts for lifting the spectrum of $H_0$ by a certain constant that depends only on the dimension, the single-site potential, the support of the random variables and the parameter $R$ of the Delone set $D$ (in \cite{EKl13}, this is proven through a different argument). These inequalities, in turn, give both the optimal uniform Wegner estimates following \cite{CHK07}, and an estimate on the spectral gap generated at the bottom of the spectrum by the finite volume restrictions. The initial length scale estimates are then a consequence of the existence of such a gap and Combes-Thomas estimates. This is done for the lower band edge in Section \ref{s:lower} and does not need a Lifshitz tails behavior. To our knowledge, so far Lifshitz tails have been obtained at the bottom of the spectrum only for the case where the unperturbed Hamiltonian is $-\Delta$ \cite{MR07,GMRM}. In the upper band edge, or in the case of a non zero background potential, it is necessary to impose a disorder condition on the probability measure of the random variables, which in a certain way replaces the Lifshitz tails estimates. Namely, one needs to impose the following condition: for some $\alpha>0$ and $\tau>d/2$
\be\label{assumption-mu} \tilde\mu((0,t])\leq \alpha t^\tau \quad \mbox{for } t>0 \mbox{ small}.\ee
where $\tilde\mu$ is the probability measure associated to the random variables $\tilde \omega(n)=M-\omega(n)$.

 In Section \ref{s:upper} we comment on the upper band edge and the geometric property needed on the Delone set $D$ to ensure that there exists spectrum around the band edge almost surely.

Using for the upper band edge of $H_\omega$ the results obtained by \cite{EKl13} for more general settings together with Theorem \ref{thm-main} one can conclude the following:
\begin{cor}\label{cor}
Let $D$ be a Delone set satisfying the property of strict positive uniform pattern frequency. Assume $\{0,M\}\in\supp\mu$ and let \eqref{assumption-mu} hold. For any $d\geq 1$, $H_\omega$ defined in \eqref{formula-ranop} exhibits dynamical localization at the spectral band edges.
\end{cor}

\bigskip

We denote by $\L_L(x)=[-L+x,L+x]^d\subset\Z^d$ the cube of side length $2L+1\in \N$ centered in $x\in\Z^d$. We denote by $H_{x,L}$ the restriction of an operator $H$ to a cube $\L_L(x)$ and omit the center $x$ from the notation when results are uniform in $x$. For a Borel set $I\subset \R$ we write $P(I)=\chi_I(H)$ for the spectral projection of $H$ associated to $I$. In particular, we use the notation $P_{0,x,L}(I)=\chi_I(H_{0,x,L})$ for the spectral projection of $H_0$ restricted to the cube $\L_L(x)$ . For $n=(n_1,...,n_d)\in\Zd$ we write $|n|_\infty:=\max_{1\leq i \leq d}\abs{n_i}$ for the max-norm and $|n|_1:=\sum_{i=1,...d}\abs{n_i}$ for the graph norm in $\ell^2(\Zd)$. We denote by $\norm{\cdot}$ the $\ell^2$-norm, and by $\ell_c(\Zd)$ the compactly supported functions on $\Zd$.

\section{Dynamical localization at the lower band edge}\label{s:lower}

\subsection{The Wegner Estimate}\label{section-WE}
The following is the uniform Wegner estimate needed to apply the Bootstrap Multiscale Analysis method for non ergodic models.
\begin{thm}\label{thm-WE} There exists an energy $E_W$ such that for any compact interval $\mathcal I\subset I=[0,E_W]$, there exists a finite constant $Q_W=Q_W(\mathcal I,\mu,R)$ such that for every $E\in \mathcal I$,
\be \Prob{\dist(\sigma(H_{\omega,x,L}), E)\leq\eta}\leq Q_W \eta L^d \ee
for $\eta\in(0,1]$ and $L\in\N$, uniformly with respect to $x\in\Zd$.
\end{thm}

Let
\be\label{maximal-pot} V_{x,L}(n)=\displaystyle \sum_{\gamma\in D\cap \L_L(x)}\delta_\gamma(n), \ee
where $\delta_\gamma(n)$ is the Kronecker delta.
The proof Theorem \ref{thm-WE} follows from the proof of \cite[Theorem 1.1]{CHK07}, which holds in the discrete setting and relies on the positivity estimate \cite[Teorem 2.1]{CHK07}. The proof of Theorem \ref{thm-WE} is therefore a consequence of the following,

\begin{lem} \label{lem-WE}
Let $L\in\N$, $L>R$. There exists an energy $E_W\in\R^+$ and a positive constant $C$ such that for $I=[0,E_W]$, for any $x\in\Z^d$ and $\varphi\in \ran P_{0,x,L}(I)$, we have
\be\label{lemma-WE-formula}  \angles{V_{x,L}\varphi,\varphi}_{\L_L(x)} \geq C \norm{\varphi}_{\L_L(x)}^2
\ee
\end{lem}

\begin{proof}
For simplicity, we omit the center of the box $x$ from the notation, since the results are uniform in $x$, and write $\norm{\cdot}$ instead of $\norm{\cdot}_{\L_L}$. Let $L> R$ and consider the spatial average of $V_L(n)$ over the cube $\L_{2R}=\{n\in\Z^d\,:\, \abs{n}_\infty\leq 2R  \}$ given by
\be W_L(n):=\frac{1}{(4R+1)^d} \sum_{\gamma\in\L_{2R}(0)}V_L(n-\gamma).  \ee
Note that for each point $n\in\L_L$, the sum in the last line is non zero. In particular, for $n\in\L_L$ near the boundary of $\L_L$, since $L>R$ there always exists a sub-cube $\L_{R}\subset \L_{2R}(n)$, completely contained in $\L_L$, that contains at least one point of $D$.
By averaging in space we retrieve the following covering condition:

\be W_L(n)\geq \frac{1}{(4R+1)^d} \sum_{m\in\L_L}\delta_m(n)= \frac{1}{(4R+1)^d} \, \chi_{\L_L}.\ee
Let $\varphi\in \ran P_{0,L}([0,1])$, with $\supp \varphi \subset \L_L$ and $\norm{\varphi}=1$, then

\begin{align}\label{eq-bound-V}
\angles{V_L\varphi,\varphi}\ & =  \angles{W_l \varphi,\varphi}+ \angles{(V_L-W_L)\varphi,\varphi} \non\\
& \geq \frac{1}{(4R+1)^d} \norm{\varphi}^2 - \angles{(W_L-V_L)\varphi,\varphi}.
\end{align}
In order to obtain a lower bound, we need an upper bound on the second term in the r.h.s. of the last inequality.

\begin{align}
\left< (W_L - V_L )\varphi, \varphi \right> & =  \frac{1}{(4R+1)^d}\displaystyle\sum_{n\in\Z^d} \overline{\varphi(n)} \left(\displaystyle\sum_{\gamma\in \L_{2R}(0)} V_L(n-\gamma)\varphi(n) \right)\non\\
& \quad \quad \quad \quad \quad\quad\quad\quad\quad \quad\quad\quad- \displaystyle\sum_{n\in\Z^d} \overline{\varphi(n)} V_L(n)\varphi(n)\non\\
& =\frac{1}{(4R+1)^d}\displaystyle\sum_{\gamma\in \L_{2R}(0)} \displaystyle\sum_{n\in\Z^d} \overline{\varphi(n+\gamma)}V_L(n)\varphi(n+\gamma) \non\\
&\quad \quad \quad \quad-  \frac{1}{(4R+1)^d}\displaystyle\sum_{\gamma\in \L_{2R}(0)} \displaystyle\sum_{n\in\Z^d} \overline{\varphi(n)} V_L(n)\varphi(n)\non\\
& =   \frac{1}{(4R+1)^d}\displaystyle\sum_{\gamma\in \L_{2R}(0)} \left(\left<V_L \varphi (\cdot +\gamma),  \varphi(\cdot +\gamma)\right> - \left< V_L\varphi, \varphi \right>   \right).
\end{align}
Note that%
\begin{align}
 \left|\left< V_L \varphi (\cdot +\gamma),\right. \right. & \left. \left. \varphi(\cdot +\gamma) \right>  -  \left<  V_L\varphi, \varphi \right>   \right| \non\\
&  =\left|\left<  V_L \varphi (\cdot +\gamma), (\varphi(\cdot +\gamma)-\varphi) \right> + \left<  V_L\varphi (\cdot+\gamma), \varphi \right>- \left<  V_L\varphi, \varphi \right>   \right| \non\\
& = \left|\left<  V_L \varphi (\cdot +\gamma), (\varphi(\cdot +\gamma)-\varphi) \right> + \left<  V_L(\varphi (\cdot+\gamma) -\varphi), \, \varphi \right>   \right| \non\\
& \leq 2\norm{V_L}_\infty \norm{\varphi(\cdot +\gamma)-\varphi}.
\end{align}

Thus
\be  |\left< (W_L - V_L) \varphi, \varphi \right> |\leq \frac{2\norm{V_L}_\infty}{(4R+1)^d}\displaystyle\sum_{\gamma\in \L_{2R}(0)}\norm{\varphi(\cdot +\gamma)-\varphi}.  \ee

For every $\gamma\in \L_{2R}(0)$, consider the shortest path between $\gamma$ and $0$ (if this is not unique, pick one).
Let $\{e_i\}_{i=1}^d$ be the canonical base of $\Zd$. For $\gamma=\sum_{i=1}^d a_i e_i$, this shortest path can be written as a sequence of vectors in $\Zd$:
\be \beta^\gamma_1:=e_1, \,  \beta^\gamma_2:=2e_1,...,  \beta^\gamma_{a_1}=a_1e_1, \ee
\be \beta^\gamma_{a_1+1}=a_1e_1+e_2, ...  , \beta^\gamma_k=a_1e_1+...+a_de_d-1, \non\ee
where $k=a_1+a_2+...a_d$. Since $\abs{\gamma}_\infty\leq 2R$, we have that $0\leq k\leq 2Rd$ and they satisfy
\be \abs{\beta^\gamma_{i+1}-\beta^\gamma_i}_1=1 \quad \mbox{for } i=0,1,2,...k, \ee
with $\beta^\gamma_0:=0$ and $\beta^\gamma_{k+1}:=\gamma$. Then,

\begin{align} \norm{\varphi(\cdot+\gamma)-\varphi}& =\norm{\varphi(\cdot+\gamma) \pm \varphi(\cdot+\beta^\gamma_k)\pm \varphi(\cdot+\beta^\gamma_{k-1})...\pm \varphi(\cdot+\beta^\gamma_1)-\varphi(\cdot) } \non\\
& \leq \norm{\varphi(\cdot+\gamma)- \varphi(\cdot+\beta^\gamma_k)} + \norm{\varphi(\cdot+\beta^\gamma_k)- \varphi(\cdot+\beta^\gamma_{k-1})} \non\\
& \hspace{4.5cm}+...+\norm{\varphi(\cdot+\beta^\gamma_1)-\varphi(\cdot)}\non.
\end{align}
This is a sum of at most $2Rd+1$ terms of the form
\be \norm{\varphi(\cdot+m)-\varphi(\cdot+n)}\quad \mbox{ where } \abs{m-n}_1=1,  \ee
and each of these terms is inferior or equal to

\be \sqrt{2\angles{H_0\varphi,\varphi}_{\L_L}}= \left( \sum_{n\in\L_L} \sum_{\abs{m-n}_1=1}\abs{ \varphi(m)-\varphi(n)}^2 \right)^{1/2}. \non\ee
Then, $\forall \gamma\in\L_{2R}(0)$, we have that $\norm{\varphi(\cdot+\gamma)-\varphi}\leq (2Rd+1)\sqrt{2\angles{H_0\varphi,\varphi}_{\L_L}}$.
This implies that
\be \sum_{\gamma\in\L_R(0)}\norm{\varphi(\cdot+\gamma)-\varphi}\leq (4R+1)^d (2Rd+1)\sqrt{2\angles{H_0\varphi,\varphi}_{\L_L}}.\ee
Then,

\be \abs{\angles{(W_L-V_L)\varphi,\varphi}} \leq 2\norm{V_L}_\infty 4\sqrt{2}Rd\sqrt{\angles{H_0\varphi,\varphi}_{\L_L}} \ee
where we used that $R\geq 1$. Replacing this in \eqref{eq-bound-V} gives
\be \angles{V_L\varphi,\varphi} \geq  \frac{1}{(4R+1)^d}- 8d\sqrt{2}\norm{V_L}_\infty R \left(\angles{H_0\varphi,\varphi}_{\L_L}\right)^{1/2}:=C  \ee
Since $\norm{V_L}_\infty=1$ and $R\geq 1$, for $C$ to be positive, it is enough that

\be \tilde E_W:=\left( (8\sqrt{2}d\, R(4R+1)^d\right)^{-2}> \angles{H_0\varphi,\varphi}_{\L_L}\ee

That is, there exists an energy $E_W:=q^2\tilde E_W>0$ with $q\in(0,1)$ such that for $I=[0,E_W]$ and $\varphi\in \ran P_{0,L}(I)$, \eqref{lemma-WE-formula} holds with a constant $C=(1-q)(4R+1)^{-d}$.
\end{proof}
\subsection{The initial length scale estimate}
To start the Bootstrap Multiscale analysis in the non ergodic setting, it is enough to prove the existence of a spectral gap above $E_0:=\inf \sigma(H_\omega)=0$ a.s. for the finite-volume operator $H_{\omega,x,L}$. This was done in the continuous setting in \cite[Proposition 3.1]{G08} using the spatial averaging argument from Section \ref{section-WE} and holds also in the discrete case in the following form

\begin{prop} For $p>0$, there exists a scale $\tilde L=\tilde L(d,\mu,p,R)$ such that for all $L\geq \tilde L$  we have

\be\label{eq-ILSE} \inf_{x\in\Z^d} \Prob{H_{\omega,x,L}\geq CR^{-2(d+2)}(\log L)^{-2/d}} \geq 1-L^{-pd}, \ee
where the positive constant $C$ depends on the parameters $p,d,M$ and $\mu$.
\end{prop}

For the reader's convenience, we sketch the proof of \cite[Proposition 3.1]{G08} with the corresponding changes in the discrete setting.

\begin{proof} For simplicity we write $\norm{\cdot}$ for $\norm{\cdot}_\LL$ and, since the results are uniform in $x$, we omit this subscript from the notation. Take $L\in\N, L>R$, and consider the spatial average of $V_{\omega,L}=V_\omega \chi_{\L_L}$ over the cube $\L_{2RK}(0)$, for some $K>1$ to be chosen later. For $n\in\L_L$, define

\begin{align} W_{\omega,L}(n) & :=\frac{1}{(4RK+1)^d} \sum_{\gamma\in\L_{2RK}(0)}V_{\omega,L}(n-\gamma)\\
& \geq \frac{1}{(5R)^d}\left( \min_{j\in \L_L\cap D}\frac{1}{K^d}\sum_{\gamma\in\L_{K/3}(j)\cap D}\omega(\gamma) \right)\chi_{\L_L}
\end{align}
We can obtain a lower bound for the last line using the theory of large deviations (see e.g. \cite[Eq. 3.10]{G08}). We obtain, for $K$ big enough

\be \Prob{W_{\omega,L}\geq \frac{1}{(5R)^d}\frac{\bar{\mu}}{2}\chi_{\LL}  }>1-L^de^{-C_{\mu,R,d} K^d}, \ee
for some constant $C_{\mu,R,d}$ depending on the probability measure $\mu$,$R$ and $d$, where $\bar{\mu}$ is the mean of $\mu$ (the dependence on $R,d$, proportional to $R^{-d}$, comes from the fact that $K^d/R^d\lesssim\abs{\L_K\cap D}\lesssim K^d$).
Therefore,

\be \bar H_{\omega,L}:=-\Delta+W_{\omega,L}\geq \frac{1}{(5R)^d}\frac{\bar \mu}{2}\quad \mbox{on }\LL, \ee
with a probability larger than $1-L^de^{C_{\mu,R,d} K^d}$.

Take $\varphi\in\ell_c(\Zd)$, $\supp\varphi\subset \LL$ and $\norm{\varphi}=1$. Then, with a probability larger than $1-L^de^{-C_{\mu,R,d} K^d}$ we have that
\begin{align}\label{eq-bound-H} \angles{H_{\omega,L}\varphi,\varphi} & =\angles{\bar H_{\omega,L}\varphi,\varphi}-\angles{(W_{\omega,L}-V_{\omega,L})\varphi,\varphi}\non\\
&\geq \frac{1}{(5R)^d}\frac{\bar{\mu}}{2}\norm{\varphi}- \angles{(W_{\omega,L}-V_{\omega,L})\varphi,\varphi}.
\end{align}
The second term in the r.h.s. of the last line can be estimated as its deterministic counterpart \eqref{eq-bound-V} in the proof of Lemma \ref{lem-WE}. Recalling that $\norm{V_\omega}_\infty=M$, we get

\be\label{eq-diff} |\left< (W_{\omega,L} - V_{\omega,L}) \varphi, \varphi \right> | \leq 8d\sqrt{2} MRK \sqrt{\angles{H_0\varphi,\varphi}_\LL} \ee

Since $V_{\omega,L}\geq 0$, $\angles{H_0\varphi,\varphi}\leq \angles{H_{\omega,L}\varphi,\varphi}$. Moreover, since we work at the bottom of the spectrum we can assume $\angles{H_{\omega,L}\varphi,\varphi}\leq 1$, by taking $\varphi\in \ran P_{\omega,L}([0,1])$. This yields
\be \angles{H_{\omega,L}\varphi,\varphi}\geq \left((5R)^d 2(1+8d\sqrt{2}MRK)\right)^{-2}(\bar \mu)^2, \ee
with a probability larger than $1-L^de^{-C_{\mu,R,d} K^d}$.

Given $p>0$, take $K=\left( \frac{(p+1)d}{C_{\mu,R,d}}\log L \right)^{1/d}$, with $L$ big enough depending on the parameters $d,\mu,p,R$, so we obtain \eqref{eq-ILSE}.

\end{proof}

\subsection{Proof of Theorem \ref{thm-main}: Localization through the Multiscale Analysis}

The Bootstrap Multiscale Analysis is at its core an iteration procedure that shows a fast decay of the local resolvents of $H_{\omega,x,L}$ in some energy interval, as $L$ tends to infinity. This in turn yields, among other things, dynamical localization. For a detailed description of the method, see \cite{K07,Kl08,S}.

As stated in \cite[Theorem 3.4]{GK01}, in order to perform the Multiscale Analysis in some energy interval $\mathcal I$ and obtain all its consequences, it is enough to verify a Wegner estimate and an initial length scale estimate (ILSE) in $\mathcal I$. This method requires some standard regularity conditions on the random potential \cite[Section 2]{GK01}. The Multiscale Analysis can be applied to non ergodic models that satisfy the aforementioned assumptions for finite-volume operators $H_{\omega,x,L}$ uniformly with respect to the center $x\in\Zd$ of the cube, see \cite[Section 2]{RM12}.

In our case, we have assumed that  the random variables are i.i.d. and have a regular probability distribution, so $V_{\omega}$ as defined in \eqref{formula-ranpot} satisfies the standard regularity assumptions. Then, what is left to verify in order to apply the Multiscale Analysis is that the two main ingredients hold, the Wegner estimate and the ILSE, uniformly with respect to the center $x\in\Zd$ \cite[Theorem 2.3]{RM12}. The first is obtained in Theorem \ref{thm-WE}, while the ILSE is a consequence of Proposition \ref{prop}.
Namely, we proved that there exists a spectral gap in $\sigma(H_{\omega,x,L})$ above the spectral infimum $E_0=0$, then the Combes-Thomas estimate transforms this into a decay of the resolvent for energies near $E_0$ (see e.g. \cite[Section 11.2]{K07} and the references therein).

\bigskip

\section{Comment on the upper band edge}\label{s:upper}

Note that by reflecting the spectrum of $H_\omega$ with respect to the origin and shifting it by a constant $4d+M$, the problem of studying the upper spectral band edge of $H_\omega$ is equivalent to study the bottom of the spectrum of the operator $\tilde H_\omega$ given by

\begin{align}\label{formula-ranoptilde} \tilde H_\omega = - H_\omega+4d+M & = \Delta+4d+M\sum_{n\in\Z^d\setminus D}\delta_n+ \sum_{n\in D} \tilde \omega(n)\delta_n\non\\
& = \tilde H_0+\tilde V_\omega
 \end{align}
 where
 \be\label{eq-defH0} \tilde H_0= \Delta+4d+M\sum_{n\in\Z^d\setminus D}\delta_n:=\Delta+4d +V_0,\ee
and $\tilde V_\omega$ is the alloy type potential of the form \eqref{formula-ranpot} corresponding to the random variables $\tilde\omega(n)=M-\omega(n)$. Denote by $\tilde \mu$ the probability measure for $\tilde \omega(n)$. We have that $\supp \tilde \mu\subset[0,M]$. Here we assume moreover the following
\be 0\in\supp \tilde \mu, \quad\mbox{that is, } M\in\supp \mu .\ee



 We need to recall some basic definitions from the theory of Delone dynamical sets, which will give us later a way of characterize the Delone sets for which working at the band edges is well defined (see \cite[Section 2.3]{MR12}).
\begin{defn}\label{defDDS}
\begin{itemize}
 \item[i)] Given an $R$-Delone set $D\subset \Z^d$, any finite subset $Q\subset D$ is called a \emph{pattern} of $D$. Two sets $Q,Q'\subset D$ are called \emph{equivalent} if there exists $x\in \Z^d$ such that $x+Q=Q'$.
\item[iii)] Let $(\L_L)_{L\in\N}$ be a sequence of concentric cubes  of side length $L$ in $\Z^d$.  We define the \emph{pattern frequency} of $Q$ in $D$ as the following limit, if it exists,
\be \eta(Q):= \lim_{L\rightarrow \infty} \frac{ \sharp \{ \tilde Q\subset D\,:\, \exists x\in (-\L_L)\, \mbox{ s.t. } x+Q=\tilde Q \} }{\abs{\L_L}}, \ee
that is, the number of equivalent patterns of $Q$ in $D$ per volume converges (it is known that is quantity is always bounded, so the question is to know whether the equality $\liminf = \limsup$ holds \cite[Lemma 2.25]{MR12}).
\item[v)] We say that $D$ has \emph{uniform pattern frequency} if for any pattern $Q\subset D$ the sequence
\be\label{upf} \eta_{x,L}(Q):= \frac{\sharp \{ \tilde Q\subset D\,:\, \exists y\in (x+\L_L) \, \mbox{ s.t. }y+\tilde Q=Q \} }{ \abs{\L_L} } \ee
converges uniformly with respect to $x\in \Z^d$, when $L$ goes to infinity.  Moreover, we say that $D$ has a \emph{strict} uniform pattern frequency (SUPF) if this limit is strictly positive.
\end{itemize}
\end{defn}

It was proven in \cite{GMRM} (see also \cite{MR07}) that if the Delone set $D$ is of finite local complexity and of strict uniform pattern frequency, then there exists a set $\Sigma\subset \R$ such that we have $\sigma(H_\omega)=\sigma(\tilde H_\omega)=\Sigma$ almost surely. What is left to prove is that $E_*=\sup \sigma(H_\omega)=\inf \sigma(\tilde H_\omega)$ is not an isolated spectral value.

Recall that for the lower spectral edge, if $0\in\supp\mu$, then one can use a Borel-Cantelli argument as in \cite[Eq. 6.4]{RM12} and follow the proof of \cite[Theorem 3.9]{K07} to show that $\sigma(-\Delta)\subset \sigma( H_\omega)$ almost surely, using the translation invariance of $-\Delta$. Note that the background potential $V_0$ in \eqref{eq-defH0} is of Delone type, since $\Z^d\setminus D$ is also a Delone set if $D$ has the SUPF property, so $\tilde H_0$ is not invariant under arbitrary translations. To recover an analogous result for $\sigma(\tilde H_\omega)$ and $\sigma(\tilde H_0)$, we can replace this translation invariance with the SUPF property. We have the following,

\begin{prop}\label{prop-H0} Let $0\in\supp\tilde\mu$ and assume that the Delone set $D$ has the property of strict uniform pattern frequency. Then, for $\tilde H_0=\Delta+4d+V_0$ defined in \eqref{eq-defH0} and $\tilde H_\omega$  defined in \eqref{formula-ranoptilde}, we have
 \be\label{formula-borel-cantelli}\sigma(\tilde H_0)\subset \sigma(\tilde H_\omega)\quad\mbox{ almost surely.} \ee
\end{prop}
\begin{proof}
Since there is a one-to-one correspondence between the patterns of $D$ and $\Z^d\setminus D$, if $D$ satisfies the SUPF property, so does $\Z^d\setminus D$. Indeed, any pattern $Q$ of $D$ is of the form $Q=D\cap K$, where $K\subset \Z^d$ is a compact set, and $\tilde Q=(\Z^d\setminus D)\cap K$ is the corresponding (unique) pattern in $\Z^d\setminus D$.

Let $E\in\sigma(\tilde H_0)$ and take a Weyl sequence $\varphi_k\in\ell_c(\Zd)$ associated to it.
For every $k$, let $K_k\subset\Zd$ be a finite cube such that $\supp\varphi_k\subset K_k$. Because of the SUPF property, every pattern $K_k\cap \Zd\setminus D$ (and therefore $K_k\cap D$) is repeated infinitely many times in $\Zd$. Fix $k$ and extract a sequence $\{ v_j(k)\}_j\subset \Zd$, with $j$ in some index set $\mathcal J(k)$, such that the cubes $\{K_k+v_j(k)\}_j$ generate patterns that are equivalent to the patterns $\{K_k\cap \Zd\setminus D\}_k$ and are pairwise disjoint.
Note that for every $j\in\mathcal J(k)$,
\be\label{formBC} V_0\chi_{K_k}=V_0\chi_{K_k+v_j(k)} \quad \mbox{and } \tilde V\chi_{K_k}=\tilde V\chi_{K_k+v_j(k)}, \ee
where $\tilde V(n)=\sum_{\gamma\in D}\delta_\gamma(n)$. This implies that the events
\be A_j=\{\omega\, : \omega_\gamma <1/k ,\,\,\forall \gamma\in K_k+v_j(k) \}  \ee
are independent and $\Prob{A_j}=\Prob{A_0}>0$, since $0\in\supp \tilde\mu$. By the Borel-Cantelli Lemma, we have that
\be \Omega_{k}:=\{\omega\,: \omega\in A_j \mbox{ for infinitely many }j\} \ee
has probability one. We have that the countable intersection $\tilde\Omega:=\bigcap_k \Omega_{k}$ has also probability one.

Let us denote the index set of the $A_j$ in $\Omega_{k}$ by $\mathcal J'(k)$. Now, among the $\{v(k)_j\}_{j\in\mathcal J'(k)}$, for every $k$ we can pick a vector $v_k \in \{v(k)_j\}_{j\in\mathcal J'(k)}$ such that the cubes $\{K_k+v_k\}_k$ are pairwise disjoint and $\norm{ V_\omega\chi_{K_k+v_k}}_\infty<1/k$ almost surely.

Then, the sequence $\{ \psi_k:=\varphi_k(\cdot-v_k)\}_k$ with $\supp\psi_k\subset K_k+v_k$, is by construction a Weyl sequence for $E$ and $\tilde H_\omega$, for $\omega\in\tilde\Omega$ (by \eqref{formBC} it is a Weyl sequence for $\tilde H_0$).
\end{proof}


 We write $E_*:=\inf \sigma(\tilde H_0)$.
 Then, since $\tilde V_\omega\geq 0$, by \eqref{formula-borel-cantelli} we have that $\inf \sigma(\tilde H_\omega)=E_*\geq 0$ a.s. Since we do not know the nature of the spectrum of $\tilde H_0$, in order to show that any dynamical localization result at the bottom of the spectrum of $\tilde H_\omega$ is non empty, we need to prove that there is almost sure spectrum in, at least, a neighborhood of $E_*$.

The following result is known in the literature of deterministic ergodic potentials in dimension 1 (see \cite[Section 6.2, Problem 2]{Su95}).

\begin{prop}\label{prop} If the Delone set $D$ has the property of strict uniform pattern frequency, then the spectrum of $\tilde H_0$ is essential spectrum.
\end{prop}
\begin{proof}
As explained in the proof of Proposition \ref{prop-H0}, if $D$ satisfies the SUPF property, so does $\Z^d\setminus D$.
Let $E\in\sigma(\tilde H_0)$. By Weyl's criterion, there exists a (normalized) Weyl sequence $\varphi_k$ in the core of $\tilde H_0$, $\ell _c(\Zd)$ such that $\norm{(\tilde H_0-E)\varphi_k}<1/k$ for any $k$ big enough. To prove that $E\in\sigma_{ess}(\tilde H_0)$, it is enough to prove that one can extract an orthogonal Weyl sequence sequence from $\varphi_k$.

For each $\varphi_k$, take a compact set $K_k\subset \Z^d$ such that $\supp \varphi_k\subset K_k$ and $\dist(\supp \varphi,\Z^d\setminus K_k)>1$. Because of the SUPF property, every pattern $K_k\cap \Zd\setminus D$ (and therefore $K_k\cap D$) is repeated infinitely many times in $\Zd$. We can extract a sequence $\{ v_k\}_k\subset \Zd$, such that the cubes $\{K_k+v_k\}_k$ generate patterns that are equivalent to the patterns \{$K_k\cap \Zd\setminus D\}_k$ and are pairwise disjoint.
Consider the sequence $\{\psi_k:=\varphi_k(\cdot-v_k)\}_k$. Since $\supp \psi_k \subset K_k+v_k$, the sequence  is orthogonal.

 Since $4d+\Delta$ is invariant with respect to translations by $v_k$ and $V_0\chi_{K_k+v_k}=V_0\chi_{K_k}$, we get $\norm{(\tilde H_0 -E)\psi_k}=\norm{(\tilde H_0-E)\varphi_k}\leq 1/k$,
i.e., $\{\psi_k\}_k$ is an orthonormal Weyl sequence for $E$.
\end{proof}

In particular, $E_*\in \sigma_{ess}(\tilde H_0)$ and $E_*\in \sigma_{ess}(\tilde H_\omega)$, by \eqref{formula-borel-cantelli}. The localization results in \cite[Section 1.2.3]{EKl13} at the bottom of the spectrum of $\tilde H_\omega$, which imply finite multiplicity of eigenvalues, together with Proposition \ref{prop} yields that $E_*$ is not an isolated spectral value of $\tilde H_\omega$. Note that in dimension $1$ this result is direct, since the essential spectrum does not contain degenerated eigenvalues.
\begin{rem}Note that property of strict uniform pattern frequency implies condition \cite[Eq. 1.13]{EKl13}.
\end{rem}

\section*{Acknowledgements}
The author would like to thank  I. Veseli\'c for providing references and M. Sabri for comments on a previous version of this article. The author is grateful to F. Germinet and A. Klein for stimulating discussions and helpful remarks. The hospitality and financial support of the Mittag-Leffler Institute (Djursholm, Sweden) are gratefully acknowledged.

\bibliographystyle{alpha}
\bibliography{refsconi_2013_01_26}

%
%
%
%

\end{document}